\documentclass[11pt]{article}

\usepackage{amsfonts}
\usepackage{booktabs}
\usepackage{amssymb,amsthm,amsmath}
\usepackage{graphicx}
\usepackage{xcolor}
\usepackage[small]{caption}
\usepackage{subcaption}
\usepackage{epsfig}

\usepackage[utf8]{inputenc}
\usepackage{fullpage}
\usepackage{a4wide}
\usepackage{framed}

\usepackage{enumerate}
\usepackage{url}
\usepackage{hyperref}
\usepackage[capitalise]{cleveref}

\theoremstyle{definition}
\newtheorem{definition}{Definition}
\newtheorem{theorem}{Theorem}
\newtheorem{corollary}{Corollary}
\newtheorem{lemma}{Lemma}

\newcounter{myclaim}[section] 
\newenvironment{myclaim}[1][]{\refstepcounter{myclaim}\par \medskip
   \noindent \textbf{Claim~\themyclaim. #1} \rmfamily}{}

\newcommand{\ie}{{i.e.}}
\newcommand{\eg}{{e.g.}}

\newcommand{\NN}{\mathbb{N}} 
\newcommand{\ZZ}{\mathbb{Z}} 
\newcommand{\RR}{\mathbb{R}} 
\newcommand{\eps}{\varepsilon}

\def\F{\mathcal F}
\def\O{\mathcal O}

\def\T{\mathcal T}

\newcommand{\Area}{\operatorname{Area}}

\def\den{\pi_T} 
\def\denL{\pi_L} 
\newcommand{\area}{A_T}
\newcommand{\areaL}{A_L}
\def\dendisc{\pi_T^{\operatorname{disc}}}
\def\areadisc{A_T^{\operatorname{disc}}}
\def\deltasup{\overline{\delta}}
\def\deltainf{\underline{\delta}}
\def\flw{\lfloor w\rfloor}
\def\flh{\lfloor h\rfloor}
\newcommand{\floor}[1]{\left\lfloor #1 \right\rfloor}
\newcommand{\ceil}[1]{\left\lceil #1 \right\rceil}
\def\covdenT{\vartheta_T}
\def\d{\operatorname{d}\!}

\newcommand{\later}[1]{}
\newcommand{\old}[1]{}

\begin{document}
\title{Piercing All Translates of a Set of Axis-Parallel Rectangles
  \footnote{A preliminary version appears in Proceedings of IWOCA 2021.}
} 

\author{
Adrian Dumitrescu\thanks{%
Algoresearch L.L.C., Milwaukee, WI, USA.\,
Email~\texttt{ad.dumitrescu@gmail.com}.}\qquad
Josef Tkadlec\thanks{%
Department of Mathematics, Harvard University, Cambridge, MA 02138, USA.
Email~\texttt{josef.tkadlec@gmail.com}.}
}

\maketitle        

\begin{abstract}
For a given shape $S$ in the plane, one can ask what is the lowest possible density
of a point set $P$ that pierces (``intersects'', ``hits'') all translates of $S$.
This is equivalent to determining the covering density of $S$ and as such is well studied.
Here we study the analogous question for families of shapes where the connection to
covering no longer exists.
That is, we require that a single point set $P$ simultaneously pierces
each translate of each shape from some family $\F$.
We denote the lowest possible density of such an $\F$-piercing point set by $\den(\F)$.
Specifically, we focus on families $\F$ consisting of axis-parallel rectangles.
When $|\F|=2$ we exactly solve the case when one rectangle is more squarish
than $2\times 1$, and give bounds (within $10\,\%$ of each other) for the remaining case
when one rectangle is wide and the other one is tall.
When $|\F|\ge 2$ we present a linear-time constant-factor approximation algorithm
for computing $\den(\F)$ (with ratio $1.895$).


\end{abstract}

\section{Introduction}  \label{sec:intro}

In a game of Battleship, the opponent secretly places ships of a fixed shape
on an $n\times n$ board and your goal is to sink them by identifying all the cells
the ships occupy (the ships are stationary). Consider now the following puzzle: 
If the opponent placed a single $2\times 3$ ship, how many attempts do you need
to surely hit the ship at least once?  
The answer depends on an extra assumption.
If you know that the ship is placed, \eg, vertically, it is fairly easy
to see that the answer is roughly $n^2/6$: When $n$ is a multiple
of~$6$, then one hit is needed per each of the $n^2/6$ interior-disjoint
translates of the $2\times 3$ rectangle that tile the board and, on
the other hand, a lattice with basis $[2,0],[0,3]$ achieves the objective. 
The starting point of this paper was to answer the question when it is
\emph{not} known whether the ship is placed vertically or horizontally.
It turns out that the answer is $n^2/5+\O(n)$ hits (the main term comes
from~\cref{thm:two}\,(ii) whereas the $\O(n)$ correction term is due to 
the boundary effect). 

Motivated by the above puzzle, we study the following problem:
Given a family $\F$ of compact shapes in the plane, what is its translative piercing density
$\den(\F)$, that is, the lowest density of a point set that pierces
(``intersects'', ``hits'') every translate of each member of the family? 
Here the density of an infinite point set $P$ (over the plane) is defined in the standard fashion
as a limit of its density over a disk $D_r$ of radius $r$, as $r$ tends to infinity. 
The piercing density $\den(\F)$ of the family is then defined as the infimum over
all point sets that pierce every translate of each member of the
family~\cite[Ch.~1]{BMP05},~\cite{FT17}.
(See~\cref{subsec:preliminaries} for precise definitions.) 
Note that unlike in the puzzle, we allow translations of each shape in the
family by any, not necessarily integer, vector. 

First, we cover the case when the family $\F=\{S\}$ consists of a single
shape. The problem is then equivalent to the
classical problem of determining the translative covering density
$\covdenT(S)$ of the shape $S$: 
Indeed, determining the translative covering density $\covdenT(S)$ amounts to
finding a (sparsest possible) point set $P$ such that the translates
$\{p+S\mid p\in P\}$ cover the plane, that is, 
\[(\forall x\in \RR^2) (\exists p\in P) \text{ such that } x\in p+S.\]
(Here ``+'' is the Minkowski sum.)
This is the same as requiring that
\[(\forall x\in \RR^2) (\exists p\in P) \text{ such that } p\in x+(-S),\]
that is, the point set $P$ pierces all translates of the shape $-S$.
Hence $\covdenT(S)=\den(\{-S\})=\den(\{S\})$.
Specifically, when $S$ tiles the plane, then the answer is simply $\den(\{S\})=1/\Area(S)$,
where $\Area(S)$ is the area of $S$.
We note that apart from the cases when $S$ tiles the plane,
the translative covering density $\covdenT(S)$ is known only for a few
special shapes $S$ such as a disk or a regular $n$-gon~\cite[Ch.~1]{BMP05}.

For the rest of this work (apart from the Conclusions) we limit ourselves
to the case when $\F$ consists of $n\ge 2$ axis-parallel rectangles.
First we consider the special case $n=2$ (\cref{thm:two} in~\cref{sec:two}),
then we consider the case of arbitrary $n\ge 2$ (\cref{thm:approx} in~\cref{sec:approx}).

\subsection{Related Work} 
There is a rich literature on related (but fundamentally different)
fronts dealing with piercing \emph{finite} collections.
One broad direction is devoted to establishing combinatorial bounds on the piercing number
as a function of other parameters of the collection, most notably the matching
number~\cite{AK92,CSZ18,CFPS15,DJ11,FdFK93,GN15,He23,HW17,Ka00,Ka91,KT96,KNPS06}
or in relation to Helly's theorem~\cite{DGK63,HB64,He23}; see also the survey
articles~\cite{Eck03,HW17}. 
Another broad direction deals with the problem of piercing a given set of
shapes in the plane (for instance axis-parallel rectangles) by the minimum number of points
and concentrates on devising algorithmic solutions, ideally exact but frequently approximate;
see for instance~\cite{CKL08,Ch03,CM05}. Indeed, the problem of computing the piercing number
corresponds to the hitting set problem in a combinatorial setting~\cite{GJ79} and is known
to be NP-hard even for the special case of axis-aligned unit squares~\cite{FPT81}.
The theory of $\eps$-nets for planar point sets and axis-parallel
rectangular ranges is yet another domain at the interface between algorithms
and combinatorics in this area~\cite{AES10,MV17}.

A third direction that appears to be most closely related to this paper is around the problem
of estimating the area of the largest empty axis-parallel rectangle amidst $n$ points
in the unit square, namely, the quantity $A(n)$ defined below. 
Given a set $S$ of $n$ points in the unit square $U=[0,1]^2$, 
a rectangle $R \subset U$ is \emph{empty} if it contains no points of $S$ in its interior.
Let $A(S)$ be the maximum volume of an empty box contained in $U$
(also known as the \emph{dispersion} of $S$), and let $A(n)$
be the minimum value of $A(S)$ over all sets $S$ of $n$ points in $U$.
It is known that $1.504 \leq \lim_{n \to \infty} n A(n) \leq 1.895$;
see also~\cite{AHR17,Kr18,So18,UV18}. The lower bound
is a recent result of Bukh and Chao~\cite{BC20} and the upper bound is another recent result of
Kritzinger and Wiart~\cite{KW21}. It is worth noting that the upper bound 
$\varphi^4/(\varphi^2+1)=1.8945\ldots$
can be expressed in terms of the golden ratio $\varphi=\frac12(1+\sqrt 5)$.
The connection will be evident in~\cref{sec:approx}.

\subsection{Preliminaries} \label{subsec:preliminaries}

Throughout this paper, a \textit{shape} is a Lebesque-measurable compact subset of the plane. 
Given a shape $S$, let $\Area(S)$ denote its area.
We identify points in the plane with the corresponding vectors from the origin.
Given two shapes $A,B\subset \RR^2$, we denote by
$A+B = \{ a+b \mid a\in A, b\in B \}$ their \textit{Minkowski sum}.
A \textit{translate} of a shape $S$ by a point (vector) $p$ is the shape
$p+S = \{ p+s \mid s\in S \}$.

In the next three definitions we introduce the (translative) piercing density
$\den(\F)$ of a family $\F$ of shapes in the plane.
Then we define a certain shorthand notation for the special case when $\F$ consists
of two axis-parallel rectangles.

\begin{definition}[$\F$-piercing sets]
  Given a family $\F$ of shapes in the plane, we say that a point set $P$ is
  \textit{$\F$-piercing} if it intersects all translates of all the shapes in $\F$, that is, if
  \[(\forall S\in \F)(\forall x\in \RR^2) (\exists p\in P) \text{ such that } p\in x+S.\]
\end{definition}

\begin{definition}[Density of a point set]
Given a point set $P$ and a bounded domain $D$ with area $\Area(D)$,
we define the \textit{density of $P$ over $D$} by
$\delta(P,D)=\frac{|P\cap D|}{\Area(D)}$. 

Given a (possibly infinite and unbounded) point set $P$, we define its
\textit{asymptotic upper and lower densities} by 
\[
\deltasup (P) =\limsup_{r\to\infty} \delta(P,D_r) 
\qquad\text{and}\qquad
\deltainf (P) =\liminf_{r\to\infty} \delta(P,D_r),
\]
where $D_r$ is the disk with radius $r$ centered at the origin.
\end{definition}

\begin{definition}[Translative piercing density $\den(\F)$]
Fix a family $\F$ of shapes in the plane.
Then we define the (translative) \textit{piercing density}
by
\[
\den(\F) = \inf_{P \text{ is $\F$-piercing}} \left\{ \deltainf(P)\right\}
\]
and the (translative) \textit{lattice piercing density} $\denL(\F)$ by
\[
\denL (\F) = \inf_{P \text{ is an $\F$-piercing lattice}} \left\{ \deltainf(P)\right\}.
\]
\end{definition}

\paragraph{Pairs of Axis-Parallel Rectangles.}
Let $R_{w\times h}$ denote a rectangle with width $w$ and height~$h$.
Here we introduce a shorthand notation for the case when $\F=\{R_{a\times b},R_{c\times d}\}$
consists of two axis-parallel rectangles.
If $a\le c$ and $b\le d$ then clearly $\den(\F)=\denL(\F)=1/(ab)$ as the lattice with basis $\{[a,0],[0,b]\}$
that pierces all translates of the smaller rectangle also pierces all translates of the larger rectangle.
Otherwise we can suppose $a\ge c$ and $b\le d$.
Stretching horizontally by a factor of $c$ and then vertically by a factor of $b$, we have
\[\den(\F)=c\cdot \den(\{ R_{\frac ac\times b},R_{1\times d}\}) =
bc\cdot\den(\{R_{\frac ac \times1},R_{1\times \frac db}\}).
\]
and likewise for $\denL(\F)$.
Thus it suffices to determine
\[
\den(w,h) := \den(\{ R_{w\times 1},R_{1\times h} \}) \quad\text{and}\quad
\denL(w,h) := \denL(\{ R_{w\times 1},R_{1\times h} \})
\]
for $w,h\ge 1$.
We say that a point set (resp. a lattice) $P$ is $(w,h)$-piercing if
it is $\{R_{w\times 1},R_{1\times h}\}$-piercing.
It is sometimes convenient to work with the reciprocals
$\area(w,h)=1/\den(w,h)$ (resp. $\areaL(w,h)=1/\denL(w,h)$)
which correspond to the \textit{largest possible per-point area}
of a $(w,h)$-piercing point set (resp. lattice).
Note that $\areaL(w,h)\le \area (w,h)$, since the sparsest
$(w,h)$-piercing point set perhaps does not have to be a lattice.
Also, $ \area (w,h)\le \min(w,h)$, as translates of the smaller rectangle
tile the plane and each translate needs to be pierced.

\subsection{Results}  \label{subsec:results}

The following theorem and its corollary summarize our results for piercing all translates
of two axis-parallel rectangles in $\RR^2$.

\begin{theorem}\label{thm:two} Fix $w,h\ge 1$.
\begin{enumerate}[{\rm (i)}]
\item When $\flw \ne \flh$ then $\area(w,h)=\areaL(w,h)=\min\{w,h\}$.
\item\label{itm:general} When $\flw = \flh=k\ge 1$, set $w=k+x$, $h=k+y$ for $x,y\in[0,1)$. Then
\[ \max\left\{ k, k+ xy - \frac{k-1}{k}(1-x)(1-y)\right\} \le \areaL(w,h)\le  \area(w,h) \le k+xy.
\]
\end{enumerate}
\end{theorem}

Note that the inequalities in (ii) become equalities in two different cases:
When $k=1$ then $\areaL(w,h)=\area(w,h)=k+xy$ and when $\min\{x,y\}=0$
(that is, when $w$ or $h$ is an integer) then $\areaL(w,h)=\area(w,h)=\min\{w,h\}=k$.

\begin{corollary} \label{cor:two}
Given a set $\F=\{R_1,R_2\}$  of two axis-parallel rectangles,
a~$1.086$\nobreakdash-approximation of $\den(\F)$ can be computed in $\O(1)$ time.
The output piercing set is a lattice with density at most $(\frac52-\sqrt 2)\cdot\den(\F)$.
\end{corollary}

\smallskip
We then address the general case of piercing all translates of any finite collection
of axis-parallel rectangles.

\begin{theorem} \label{thm:approx}
Given a family $\F=\{R_1,\dots,R_n\}$ consisting of $n$ axis-parallel rectangles,
a~$1.895$-approximation of $\den(\F)$ can be computed in $\O(n)$ time.
The output piercing set is a lattice with density at most $(1+\frac25\sqrt 5)\cdot\den(\F)$.
\end{theorem}

\section{Piercing Two Rectangles}  \label{sec:two}
In this section we prove~\cref{thm:two} and deduce~\cref{cor:two}.
The most involved part of the proof is the upper bound $\area(w,h)\le k+xy$.
The proof is an integral calculus argument, which originates from a probabilistic argument. 

\begin{proof}[Proof of~\cref{thm:two}]
(i) Note that $\area(w,h) \le \min\{w,h\}$:
  Indeed, any $(w,h)$-piercing point set has to pierce all
  the translates of the rectangle with smaller area and certain copies
  of that smaller rectangle tile the plane. To complete the proof,
  it suffices to exhibit a suitable $(w,h)$-piercing lattice.
Without loss of generality suppose that $\flh < \flw$. 
We will show that the lattice $\Lambda_1$ with basis $u_1=[1,h-1]$, $v_1=[1,-1]$ 
(see~\cref{fig:lattices}(a)) is $(w,h)$-piercing. Note that the area of the fundamental
parallelogram of the lattice is $(h-1)+1=h$, as required.

We first show that $\Lambda_1$ pierces all $1 \times h$ rectangles. Observe that the $1 \times h$
rectangles centered at points in $\Lambda_1$ tile the plane. Denote this tiling by $\T$.
Let now $R$ be any $1 \times h$ rectangle. Its center is contained in one of the
rectangles in $\T$, say $\sigma$. Then the center of $\sigma$ pierces $R$, as required.

We next show that $\Lambda_1$ pierces all $w \times 1$ rectangles.
It suffices to show that $\Lambda_1$ pierces all $w_0 \times 1$ rectangles, where $w_0=\flh +1$.
Let $R$ be any $w_0 \times 1$ rectangle. Assume that $R$ is not pierced by $\Lambda_1$.
Translate $R$ downwards until it hits a point in $\Lambda_1$, say $q$, and then leftwards
until it hits another point in $\Lambda_1$, say $p$. Let $R'$ denote the resulting rectangle.
Then $p$ is the top left corner of $R'$. Observe that the top and the right side of $R'$
are not incident to any other point in $\Lambda_1$. Consider the lattice point $s:=p+u_1+(w_0-1)v_1$;
note that $x(s)-x(p)=w_0$ and $y(s)-y(p)=h-1-\flh \in [-1,0)$.
As such, $s$ is contained in the right side of $R'$, a contradiction.
It follows that $R$ is pierced by $\Lambda_1$, as required. 

\begin{figure}[ht]
\centering
\includegraphics[width=\textwidth]{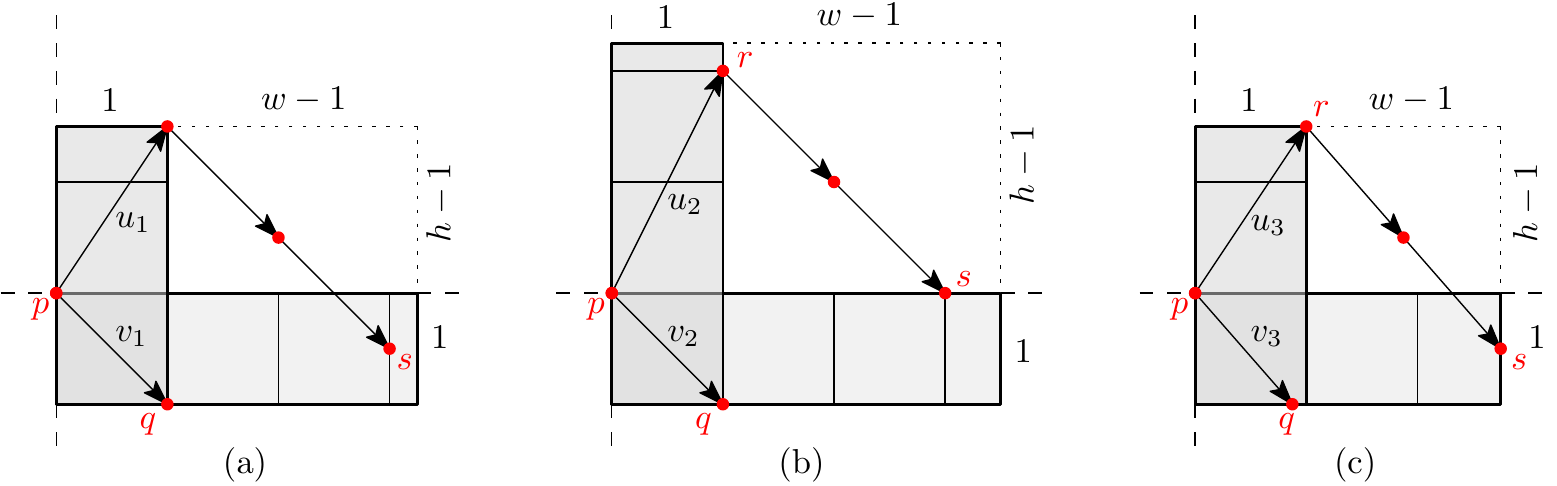}
\caption{(a) A lattice $\Lambda_1$ for the case $\flw\ne\flh$. Here $w=3+\frac14$, $h=2+\frac12$.
  (b) A lattice $\Lambda_2$ with basis $u_2=[1,k-1]$, $v_2=[1,-1]$  attesting that
  $\area(k+x,k+y)\ge k$. Here $w=3+\frac12$, $h=3+\frac14$. 
  (c) A lattice $\Lambda_3$ with basis $u_3=[1,h-1]$, $v_3=[(w-1)/k,-1]$ attesting that
  $\area(k+x,k+y)\ge  k+ xy - \frac{k-1}{k}(1-x)(1-y)$. Here $w=2+\frac34$, $h=2+\frac12$.
}
\label{fig:lattices}
\end{figure}

(ii) In order to prove the lower bound it suffices to exhibit suitable
lattices. We will show that the following lattices do the job: 
The lattice $\Lambda_2$ with basis $u_2=[1,k-1]$, $v_2=[1,-1]$ 
(see~\cref{fig:lattices}(b)) attests that $\area(k+x,k+y)\ge k$.
Note that the area of the fundamental parallelogram of $\Lambda_2$ is
$(k-1)+1=k$, as required.
The lattice $\Lambda_3$ with basis $u_3=[1,h-1]$, $v_3=[(w-1)/k,-1]$
(see~\cref{fig:lattices}(c)) attests that
$\area(k+x,k+y)\ge k+ xy - \frac{k-1}{k}(1-x)(1-y)$.
Note that the area of the fundamental parallelogram of $\Lambda_3$ is
$(w-1)(h-1)/k +1= (k+ xy) - \frac{k-1}{k}(1-x)(1-y)$, as required.

For both lattices, the proof proceeds by contradiction as in part~(i).
Assume that there exists an unpierced rectangle of dimensions
either $w\times 1$ or $1\times h$.
Translate the rectangle downwards until it hits a point in the lattice, say $q$,
and then leftwards until it hits another point in the lattice, say $p$.
For $\Lambda_2$, note that
$r:=p+u_2$ lies on the right edge of the $1\times h$ rectangle and that
$s:=p+u_2+(k-1)v_2$ lies on the top edge of the $w\times 1$ rectangle.
Similarly, for $\Lambda_3$ note that
$r:=p+u_3$ is the top right corner of the $1\times h$ rectangle and that
$s:=p+u_3+ k v_3$ lies on the right edge of the $w\times 1$ rectangle.   
Either way, we get a contradiction.

\smallskip
Finally, we show the upper bound, that is, $\area(w,h)\le k+xy$.
Recall that $\area(w,h) \le \min\{w,h\}=k + \min\{x,y\}$; we will obtain an improved
bound $\area(w,h)\le k+xy$ by an integral calculus argument. 
Let $P$ be a $(w,h)$-piercing point set, where $w=k+x$, $h=k+y$ with 
$k\in\NN$ and $x,y\in [0,1)$. 
The desired upper bound on $\area(w,h)$ will follow from a lower bound on the density
$ \delta(P,D_r) = \frac{|P \cap D_r|}{\Area(D_r)}$,
where $D_r$ is the disk with radius $r$ centered at the origin.
Fix a radius $r$ and write $P_r = P \cap D_r$.

Given a point $a=(a_x,a_y)\in\RR^2$, we denote by $R_a=[a_x-w,a_x]\times [a_y-h,a_y]$
the $w\times h$ rectangle whose top right corner is $a$.
For brevity, we denote $R=R_{(w,h)}=[0,w]\times [0,h]$.
We consider two sets of $w\times h$ rectangles:
Those that intersect $D_r$ and those that are contained within $D_r$.
We denote the sets of their top right corners by
$X=\{a\in \RR^2\mid R_a\cap D_r\neq\emptyset\}$ and
$W=\{a\in\RR^2\mid R_a\subset D_r\}$, respectively.
See~\cref{fig:proof-two}(a).

\begin{figure}[ht]
\centering
\includegraphics[width=0.95\textwidth]{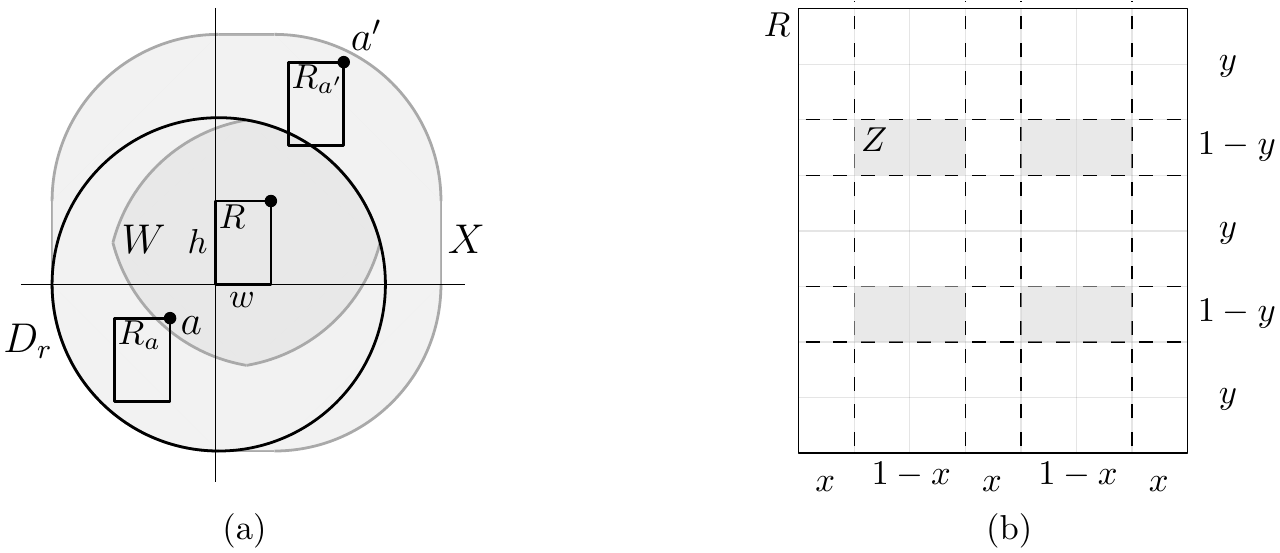}
\caption{(a) The top right corners of rectangles intersecting $D_r$ form a region $X=D_r+R$.
The top right corners of rectangles contained within $D_r$ form a region
$W=D_r\cap ([w,0]+D_r) \cap ([0,h]+D_r) \cap ([w,h]+D_r)$.
Both regions are convex and their boundaries consist of circular arcs and line segments.
(b) A $w\times h$ rectangle $R=R_{(w,h)}$ (here $k=2$, $x=1/3$, and $y=2/3$,
  hence $w=k+x=7/3$ and $h=k+y=8/3$). Its zone $Z$ is shaded.}
\label{fig:proof-two}
\end{figure}

Given a rectangle $R_a$, we define its \textit{zone} $Z_a$ to be a
union of $k^2$ closed rectangles with sizes $(1-x)\times (1-y)$ each, arranged as
in~\cref{fig:proof-two}(b). Note that $\Area(Z_a)=k^2(1-x)(1-y)$.
Further, let $I_a:=|P_r \cap Z_a|$ and $J_a:=|P_r \cap (R_a \setminus Z_a)|$
be the number of points of $P_r$ contained in $R_a$ inside its zone and outside of it, respectively.
We make two claims about $I_a$ and $J_a$.

\begin{myclaim}\label{claim1} 
If $a\in W$ then $(k+1)I_a+ kJ_a \geq k(k+1)$.
\end{myclaim}
\begin{proof}
Fix $a\in W$. Since $R_a\subset D_r$, we have $P\cap R_a=P_r\cap R_a$.
The key observation is that for any point $p \in R_a \setminus Z_a$, the set
$R_a \setminus \{p\}$ contains $k$ pairwise disjoint rectangles of
dimensions either all $w\times 1$ or all $1\times h$. We thus must have
$|P_r \cap R_a |\ge k+1$, except when $P_r \cap R_a \subseteq Z_a$, in which
case we must have $| P_r\cap R_a| \ge k$. 

Denote $I=I_a$ and $J=J_a$. There are two simple cases:
\begin{enumerate}
\item $J \geq 1$: Then $I+J \geq k+1$, thus $(k+1)I+ kJ \ge kI+kJ \geq k(k+1)$.
\item $J=0$: Then $I \ge k$, thus $(k+1)I+ kJ \ge k(k+1)$. \qedhere
\end{enumerate}
\end{proof}

\begin{myclaim} \label{claim2}
We have
\[ \int_X \frac{I_a}{\Area(Z)} \d a= \int_X \frac{J_a}{\Area(R\setminus Z)} \d a = |P_r|.
\]
\end{myclaim}
\begin{proof}
Fix $p\in P_r$. Note that the set $X_p=\{a\in\RR^2\mid p\in Z_a\}$ of top right corners of
$w\times h$ rectangles whose zone contains $p$ is a subset of $X$ congruent to $Z$.
Thus $\Area(X_p)=\Area(Z)$. Summing over $p\in P_r$ we obtain
\[\int_X \frac{I_a}{\Area(Z)} \d a= \frac{\sum_{p\in P_r}\Area(X_p)}{\Area(Z)} = |P_r|.
\]
For $J_a$ we proceed completely analogously.
\qedhere
\end{proof}

Now we put the two claims together to get a lower bound on $|P_r|$.

\begin{myclaim} \label{claim3}
We have $ |P_r| \geq \frac{\Area(W)}{k+xy}.$
\end{myclaim}
\begin{proof}
First, applying Claim~\ref{claim1} to all $w\times h$ rectangles $R_a$ with $a\in W$ and then invoking
$W\subset X$, we obtain
\begin{align*}
\Area(W)\cdot k(k+1) &= \int_{W} k(k+1) \d a \le (k+1)\int_W I_a \d a+k\int_W J_a \d a\\
	&\le (k+1)\int_X I_a\d a + k\int_X J_a\d a.
\end{align*}
By Claim~\ref{claim2} and straightforward algebra we further rewrite this as
\begin{align*}
  (k+1)\int_X I_a\d a &+ k\int_X J_a\d a =
  |P_r|\cdot \big( (k+1)\Area(Z) + k\Area(R\setminus Z)\big) \\
&= |P_r|\cdot (k\Area(R) + \Area(Z)) = |P_r|\cdot k(k+1)(k+xy),
\end{align*}
where the last equality follows from
\[ k \Area(R) + \Area(Z) = k(k+x)(k+y) + k^2(1-x)(1-y) = k(k+1)(k+xy).
\]
The bound $|P_r| \geq \frac{\Area(W)}{k+xy}$ follows by rearranging.
\qedhere
\end{proof}

Consequently, by Claim~\ref{claim3} we have
\[
  \delta(P,D_r) = \frac{|P_r|}{\Area(D_r)} \geq
\frac{\Area(W)}{\Area(D_r)} \cdot \frac1{k+xy}
 \to_{r\to\infty} \frac{1}{k+xy},
\]
where we used that $\Area(W) / \Area(D_r)\to 1$ as $r\to\infty$.
This in turn gives
\[\den(w,h) = \inf_{P}\left\{ \liminf_{r\to\infty}  \delta(P,D_r) \right\}\geq \frac{1}{k+xy}
\text{ and } \area(w,h) =\frac1{\den(w,h)} \le k+xy
\]
and completes the proof of Theorem~\ref{thm:two}.
\qedhere
\end{proof}

For a visual illustration of our results, see~\cref{fig:plots}. 

\begin{figure}[ht]
\centering
\includegraphics[width=0.95\textwidth]{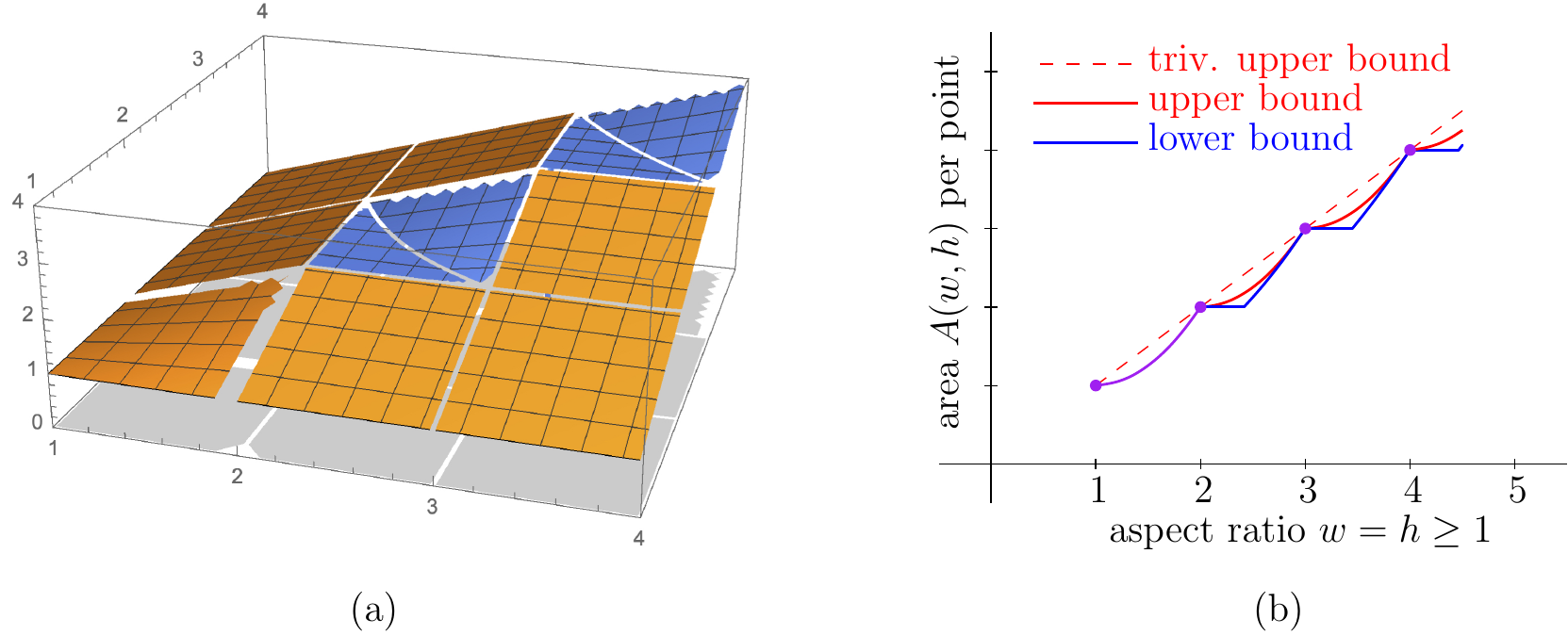}
\caption{
(a) We plot $\area(w,h)$ when $\flw\ne\flh$ and/or when
$\flw=\flh=1$ (orange). When $\flw=\flh \ge 2$ we plot the two lower
bounds from~\cref{thm:two}, \cref{itm:general} (blue). As
$k\to\infty$, the two lower bounds coincide for $x+y=1$. 
(b) A section corresponding to $w=h$. We plot the lower bounds
(blue) and the upper bound (red) on $\area(w,w)$
from~\cref{thm:two},~\cref{itm:general} and the trivial upper bound
$\area(w,w)\le w$ (red, dashed).
} 
\label{fig:plots}
\end{figure}

\begin{proof}[Proof of~\cref{cor:two}]
  It suffices to show that
  \[ \sup_{\stackrel{k \geq 2, \, k \in \NN}{x,y \in [0,1)}}
\frac {k+xy} {\max\left\{ k, k+ xy - \frac{k-1}{k}(1-x)(1-y)\right\}}
=\frac{5-2\sqrt 2}2< 1.086. \]
A computer algebra system (such as Mathematica) shows that the supremum is attained
when $k=2$ and when $x$, $y$ are both equal to a value that makes the two expressions
inside the $\max\{\}$ operator equal. This happens for $x=y=\sqrt 2-1$ and the corresponding
value is $(5-2\sqrt 2)/2 <1.086$ as claimed.
\qedhere
\end{proof}

\section{Piercing $n$ Rectangles}  \label{sec:approx}

In this section we prove~\cref{thm:approx}.
Let $\varphi=\frac12(1+\sqrt 5)$ be the golden ratio.
In~\cref{lem:area} we show that a lattice $\Lambda_\varphi$
with basis $u=[1,\varphi], v=[\varphi,-1]$ pierces all rectangles with area $\varphi^4$ or larger,
irrespective of their aspect ratio.  See~\cref{fig:phi-lattice}\,(a).
\cref{thm:approx} then follows easily by rescaling $\Lambda_\varphi$ to match
the smallest-area rectangle from the family.   

\begin{figure}[ht]
\centering
\includegraphics[width=0.95\textwidth]{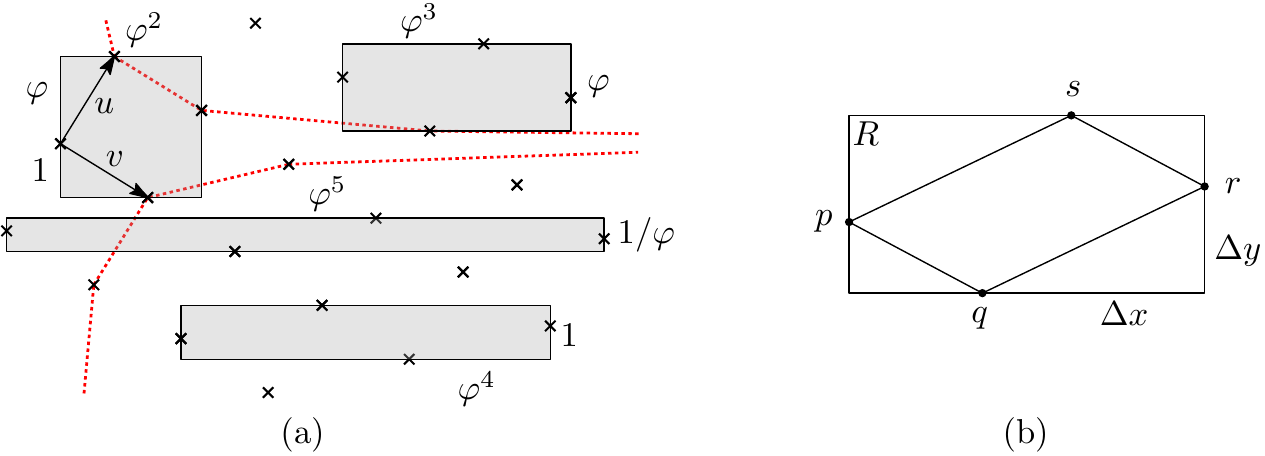}
\caption{(a) Empty rectangles amidst $\Lambda_\varphi$.
(b) A generic empty rectangle $R$.}
\label{fig:phi-lattice}
\end{figure}

Recall the well-known sequence of Fibonacci numbers defined by the following recurrence:
\begin{equation}
F_i = F_{i-1} + F_{i-2}, \text{ with } F_1=F_2=1.
\end{equation}
The first few terms in the sequence are listed in Table~\ref{tab:1} for easy reference;
here it is convenient to extend this sequence by $F_{-1}=1$ and $F_0=0$.

\begin{table}[htb]
\caption{The first few Fibonacci numbers.} 
\label{tab:1}
\centering
\begin{tabular}{|c||c|c|c|c|c|c|c|c|c|c|c|c|}
\hline
$m$ & $-1$ & 0 & 1 & 2 & 3 & 4 & 5 & 6 & 7 & 8 & 9 & 10 \\ 
\hline
$F_m$  & 1 & 0 & 1 & 1 & 2 & 3 & 5 & 8 & 13 & 21 & 34 & 55 \\
\hline
\end{tabular}
\end{table}

We first list several properties of Fibonacci numbers.
\begin{lemma} \label{lem:identities}
The following identities hold for every integer $m \geq 1$:
\begin{enumerate}
\item\label{itm:positive} $F_m \varphi + F_{m-1} = \varphi^m$,
\item\label{itm:negative} $F_m \varphi - F_{m+1} = (-1)^{m+1} \varphi^{-m}$,
\item\label{itm:square} $F_{2m+1}F_{2m-1} - (F_{2m})^2 =1$.
\end{enumerate}
\end{lemma}
\begin{proof} 
  This is straightforward to verify, for instance using the well-known formula
$F_m=\frac1{\sqrt5}(\varphi^m-\psi^m)$, where $\psi=-1/\varphi$.
\qedhere
\end{proof}

Next we prove a lemma that establishes a key property of the lattice $\Lambda_\varphi$.

\begin{lemma} \label{lem:area}
  The area of every empty rectangle amidst the points in $\Lambda_\varphi$ is
  at most~$\varphi^4$.  
\end{lemma}
\begin{proof}
Let $R$ be a maximal axis-parallel empty rectangle.
Note that since $\varphi$ is irrational, $R$ is bounded.
Denote by $p$, $q$, $r$, $s$ the four lattice points that bound $R$
from the left, below, right and top, respectively.
Refer to~\cref{fig:phi-lattice}\,(b).
Then $pqrs$ is a fundamental parallelogram of the lattice;
as such, its area is $\varphi^2+1$.
Clearly, we may assume $p=(0,0)$.
Further, we may assume $q=cu+dv$, $s=au +bv$,
where $a$, $b$, $c$, $d$ are nonnegative integers:
Indeed, since $\Lambda_\varphi$ is invariant under rotation by $90^\circ$, we can assume that
the width of $R$ is at least as large as its height. Points $q$, $s$ thus lie on the ``funnel''
(depicted in \cref{fig:phi-lattice}\,(a) dotted) within the angle formed by the vectors $u$, $v$. 
The coordinates of points $s$, $q$, $r$ and the area of the parallelogram $pqrs$ are:
\begin{align*}
s &= (a+b\varphi, a\varphi -b), \\
q &= (c+d\varphi, c\varphi -d), \\
r &= ((a+c)+(b+d)\varphi, (a+c)\varphi, -(b+d)), \\
\Area(pqrs) &= |(a+b\varphi)(c\varphi -d) - (a\varphi -b)(c+d\varphi)|=
|(ad-bc)| (\varphi^2+1).
\end{align*}

Since $s$ lies above the horizontal line through $p$ and since $q$ lies below it, we have
$a \varphi -b>0$ and $c \varphi -d <0$.
This implies $a,b,c,d>0$ and 
rewrites as $\frac ba<\varphi <\frac dc$, so in particular $ad>bc$.
Together with the expression for $\Area(pqrs)=\varphi^2+1$ this yields $|ad-bc|=ad-bc=1$.
To summarize, we have
\begin{equation} \label{eq:phi-fractions}
\frac{b}{a} < \varphi < \frac{d}{c} \quad\text{and}\quad ad -bc=1.
\end{equation}

The relation $ad-bc=1$ implies that $\gcd(a,b)=\gcd(c,d)=1$.
By a result from the theory of continued fractions~\cite{Olds63},\cite[Ch.~10]{HW79},
relation~\eqref{eq:phi-fractions} implies that the fractions $b/a$ and $d/c$
are consecutive convergents of $\varphi$. Moreover, it is well known that
the convergents of $\varphi$ are ratios of consecutive Fibonacci numbers:
\begin{align*}
\frac{F_0}{F_{-1}} < \frac{F_2}{F_1} < \frac{F_4}{F_3} < \frac{F_6}{F_5} < 
\frac{F_8}{F_7} < \cdots &<\varphi
< \cdots  < \frac{F_7}{F_6} < \frac{F_5}{F_4} < \frac{F_3}{F_2}
< \frac{F_1}{F_0} =\infty, \\
\frac{0}{1} < \frac{1}{1} < \frac{3}{2} < \frac{8}{5} < 
\frac{21}{13} < \cdots &<\varphi
< \cdots < \frac{13}{8} < \frac{5}{3} < \frac{2}{1} 
< \frac{1}{0} =\infty. 
\end{align*}

Thus we can assume that $\frac{b}{a}= \frac{F_{2k}}{F_{2k-1}}$.
There are two cases:

\emph{Case 1:} $\frac{d}{c}= \frac{F_{2k-1}}{F_{2k-2}}$, and thus
$\frac{F_{2k}}{F_{2k-1}} < \varphi < \frac{F_{2k-1}}{F_{2k-2}}$.

\emph{Case 2:} $\frac{d}{c}= \frac{F_{2k+1}}{F_{2k}}$, and thus
$\frac{F_{2k}}{F_{2k-1}} < \varphi < \frac{F_{2k+1}}{F_{2k}}$. 

\smallskip
(Note that in both cases we indeed have $ad-bc=1$ by~\cref{lem:identities},~\cref{itm:square}.)
We compute $\Area(R)$ using~\cref{itm:positive,itm:negative} of~\cref{lem:identities}.
Let $\Delta{x}$ and $\Delta{y}$ denote the side-lengths of $R$.

\smallskip
\noindent\emph{Case 1:} We have
\begin{align*}
  \Delta{x} &=  (a+c)+(b+d)\varphi = (F_{2k-1}+ F_{2k-2}) + (F_{2k}+F_{2k-1})\varphi\\
  &=
  F_{2k} + F_{2k+1}\varphi = \varphi^{2k+1}, \\
 \Delta{y} &= (a-c)\varphi - (b-d) = (F_{2k-1} - F_{2k-2})\varphi - (F_{2k}-F_{2k-1})\\
 &=
 F_{2k-3}\varphi  - F_{2k-2} = \varphi^{-(2k-3)},
\end{align*}
thus $\Area(R) =  \Delta{x} \cdot  \Delta{y} = \varphi^{2k+1} \varphi^{-(2k-3)} = \varphi^4$,
as required.

\smallskip
\noindent \emph{Case 2:} Similarly, we have
\begin{align*}
  \Delta{x} &=  (a+c)+(b+d)\varphi = (F_{2k-1}+ F_{2k}) + (F_{2k}+F_{2k+1})\varphi\\
  &=
  F_{2k+1} + F_{2k+2}\varphi = \varphi^{2k+2}, \\
 \Delta{y} &= (a-c)\varphi - (b-d) = (F_{2k-1} - F_{2k})\varphi - (F_{2k}-F_{2k+1})\\
 &=
 -F_{2k-2}\varphi  + F_{2k-1} = \varphi^{-(2k-2)},
\end{align*}
thus $\Area(R) =  \Delta{x} \cdot  \Delta{y} = \varphi^{2k+2} \varphi^{-(2k-2)} = \varphi^4$,
as required.
\qedhere
\end{proof}

\smallskip
With~\cref{lem:area} at hand, the proof of~\cref{thm:approx} is straightforward.
\begin{proof}[Proof of~\cref{thm:approx}] Let $R\in\F$ be the smallest-area rectangle among those in $\F$.
  By~\cref{lem:area}, the lattice $\Lambda_\varphi$ pierces all axis-parallel rectangles
  with area at least $\varphi^4$.
Thus the rescaled lattice $ \Lambda_\varphi'=\sqrt{\Area(R)/\varphi^4}\cdot \Lambda_\varphi$ pierces
all axis-parallel rectangles with area at least $\Area(R)$. In particular, it pierces all rectangles
in $\F$. Since the fundamental parallelogram of $\Lambda_\varphi$ has area $\varphi^2+1$,
the fundamental parallelogram of $\Lambda_\varphi'$ has area $(\varphi^2+1)/\varphi^4\cdot \Area(R)$
and gives an approximation factor $\varphi^4/(\varphi^2+1) =1+\frac 25\sqrt 5<1.895$ as claimed.
Note that computing the smallest-area rectangle and the rescaling only take $\O(n)$ time.
\qedhere
\end{proof}

\paragraph{Remarks.} We have learned from the recent article of Kritzinger and Wiart~\cite{KW21} that
a rescaled version of the lattice $\Lambda_\varphi$ was considered several years ago by Thomas Lachmann
(unpublished result) as a candidate for an upper bound on the minimum dispersion $A(n)$ of an
$n$-point set in a unit square. Yet another lattice resembling $\Lambda_\varphi$ was studied
in the same context by Ism\u{a}ilescu~\cite{Is19}. 

It is easy to check that the lattice $\Lambda_\varphi$ yields the upper bound
\[\liminf_{n \to \infty} n A(n) \leq \varphi^4/(\varphi^2+1),
\] 
\ie, matching exactly the dispersion bound obtained by Kritzinger and Wiart using a suitable modification
of the so-called Fibonacci lattice~\cite{FS89}. It is worth noting that:
(i)~the lattice $\Lambda_\varphi$ yields the above dispersion result with a cleaner and shorter proof;
(ii)~the Fibonacci lattice as well as its modification lead to this bound only by a limiting process; 
and perhaps more importantly,
(iii)~the upper bound in~\cref{lem:area} on the maximum rectangle area amidst points in this lattice
holds universally across the entire plane and not only inside a bounding box with $n$ points
(\ie, one does not need to worry about rectangles with a side supported by the bounding box boundary).

\section{Conclusion}  \label{sec:conclusion}

We list several open questions.

\smallskip
\begin{enumerate} \itemsep 3pt

\item (Computational complexity.) Given a family $\F$ of $n$ axis-parallel rectangles,
  what is the computational complexity of determining the optimal density $\den(\F)$
  of a piercing set for all translates of all members in $\F$?
  (For the decision problem: given a threshold $\tau$,
  is there a piercing set whose density is at most $\tau$?)
  Is the problem algorithmically solvable?
  Is there a polynomial-time algorithm?
  And how about the complexity of determining the optimal density $\denL(\F)$ of
  a piercing lattice for $\F$?

\item (Exact answer for two rectangles.) For two rectangles, what is the actual value
  of $\den(w,h)$ (or its reciprocal $\area(w,h)$)
  when $\flw=\flh \ge 2$? Is it the same as $\denL(w,h)$? 

\item (Other shapes.) How about pairs of different shapes such as triangles?

\item (Congruent copies.) How about requiring that the point set pierces all congruent copies
  of a set of shapes, not only translates? See for instance~\cite{BS18}
  where it is shown that when $\F$ consists of a single square (or rectangle),
  a suitable triangular grid gives an upper bound on the density.

\item (Discrete version.) Consider a discrete version where:
(i)~each shape in the family $\F$ consists of cells of an infinite square grid;
(ii)~we consider translates by integer vectors only; and
(iii)~instead of piercing with a point set we pierce with a set of grid cells.
What can be said about this (lowest possible) discrete hitting density $\dendisc(\F)$
or its reciprocal $\areadisc(\F)$?
As one example, we note that Theorems~\ref{thm:two} and~\ref{thm:approx} can be adapted
to the discrete setting in a straightforward way:
When $\F$ consists of two rectangles $R_{a\times b}$, $R_{b\times a}$, where
$b=k\cdot a+r$ (with $k\ge 1$ and $r<a$), the adapted~\cref{thm:two} yields an upper bound
$\areadisc(\{R_{a\times b},R_{b\times a}\})\le k\cdot a^2+r^2$
that solves the puzzle we mentioned in the introduction.
Furthermore, when $\F$ consists of all rectangles that have area at least $K$,
then the adapted~\cref{thm:approx} gives an $\F$-piercing set of cells
with density $(1+\frac25\sqrt{5})/K$.
In the language of Fiat and Shamir~\cite{FS89}, there is a probing strategy that locates
a battleship of $K$ squares in a rectangular sea of $M$ squares (where $M \to \infty$)
in at most $1.895 \, M/K$ probes. This is a substantial improvement over the
$3.065 \, M/K$ bound from~\cite{FS89}.

\smallskip
As another example, when $\F$ consists of two $L$-triominoes
that are centrally symmetric to each other, one can show that $\areadisc(\F)=3$.
In contrast, when $\F$ consists of two (or more) $L$-triominoes,
one of which is obtained from another one by rotation by $90^\circ$,
one can show that $\areadisc(\F)=2$.

\item (Higher dimensions.) Consider the problem in higher dimensions.
When $d=3$ and $|\F|=2$, stretching along the coordinate axis yields
a non-trivial case $\{(a,b,1),(1,1,c)\}$ where $a,b,c \geq 1$.
The trivial upper bound $\area(a,b,c)\le \min\{ab,c\}$
can be matched in two ``easy'' cases:

(i) When $c\ge \ceil{a}\cdot\ceil{b}$ then ``piercing $1\times 1\times c$ is free'':
Briefly, in the plane $z=0$ we use a lattice with basis $[a,0],[0,b]$ and in the planes $z\in\ZZ$
we consider $\ceil{a}\cdot\ceil{b}$ ``integer horizontal offsets'' $[u,v]$, $u=0,\dots,\ceil{a}-1$,
$v=0,\dots,\ceil{b}-1$ and use them periodically (in any order).
    
(ii) When $c\le \floor{a}\cdot\floor{b}$ then ``piercing $a\times b\times 1$ is free'':
Briefly, along the line $x=y=0$ we put points at $k\cdot  \floor{a}\cdot\floor{b}$ for $k\in\ZZ$.
For other vertical lines through integer points we use $ \floor{a}\cdot\floor{b}$
``integer vertical offsets'' such that every $ \floor{a}\cdot\floor{b}$ horizontal
grid-rectangle contains all offsets.
    
Together, the easy cases (i) and (ii) cover the case when $a\in\ZZ$, $b\in\ZZ$, $c\in\RR$
and the case when $ab$ and $c$ ``differ by a lot''. 
Another easy case is $a=1$, when the planar bounds apply (for two rectangles with sizes
$1\times b$ and $c\times 1$).

\smallskip
Finally, from the algorithmic standpoint, given a finite collection of axis-parallel boxes
in $\RR^d$, what approximations for the piercing density can be obtained?

\item (Disconnected shapes.) What can be said about disconnected shapes? The easiest variant
  seems to be when each shape is a set of integer points on the line (or in $\ZZ^d$)
  and we consider translates by integer vectors only.
  Note that the piercing density and the lattice piercing densities may differ in this case;
  for instance, when $S=\{0,2\} \subset \ZZ$, these densities are $1/2$ and $1$, respectively.
  In view of the connection to covering mentioned in~\cref{sec:intro}, it is worth
  mentioning that the problem of tiling the infinite integer grid with finite clusters is
  only partially solved~\cite{Sz98}; however, covering is generally easier than tiling. 
\end{enumerate}

\section*{Acknowledgments} The authors thank Wolfgang Mulzer and Jakub Svoboda
for helpful comments on an earlier version of this work.

\end{document}